\documentclass[a4paper, 10pt]{article}

\usepackage{algorithm}
\usepackage{algpseudocode}
\usepackage{verbatim}
\usepackage[cmex10]{amsmath}
\usepackage{amsfonts}

\newtheorem{lemma}{Lemma}

\newtheorem{theorem}[lemma]{Theorem}
\newtheorem{definition}[lemma]{Definition}

\newtheorem{corollary}[lemma]{Corollary}
\newcommand{\E}{\mbox{\rm E}}

\newcommand{\qed}{\hfill\ensuremath{\Box}\medskip\\\noindent}
\newenvironment{proof}{\noindent\emph{Proof. }}{\qed}

\usepackage{tikz}

\usepackage{times}

\usepackage{cite}
\usepackage{array}

\usepackage{mdwmath}
\usepackage{mdwtab}

\usepackage[tight,footnotesize]{subfigure}

\begin{document}
\title{On Finding Similar Items in a Stream of Transactions}


\author{Andrea Campagna and Rasmus Pagh\\
IT University of Copenhagen, Denmark\\
Email: {\tt \{acam,pagh\}@itu.dk}}
\maketitle

\begin{abstract}
	While there has been a lot of work on finding frequent itemsets in transaction data streams, none of these solve the problem of finding similar pairs according to standard similarity measures.
	This paper is a first attempt at dealing with this, arguably more important, problem.

	We start out with a negative result that also explains the lack of theoretical upper bounds on the space usage of data mining algorithms for finding frequent itemsets:
	Any algorithm that (even only approximately and with a chance of error) finds the most frequent $k$-itemset must use space $\Omega(\min\{mb,n^k,(mb/\varphi)^k\})$ bits, where $mb$ is the number of items in the stream so far, $n$ is the number of distinct items and $\varphi$ is a support threshold. 
	
	To achieve any non-trivial space upper bound we must thus abandon a worst-case assumption on the data stream.
	We work under the model that the transactions come in random order, and show that surprisingly, not only is small-space similarity mining possible for the most common similarity measures, but the mining accuracy {\em improves\/} with the length of the stream for any fixed support threshold.
\end{abstract}

\section{Introduction}
Imagine that we have a set of $m$ sets (``transactions''), each a subset of $\{1,\ldots,n\}$, and that we want
to find interesting associations among items in these transactions.
This problem is often framed in a ``market basket'' model where
we are interested in finding those pairs of items that are frequently
bought together.
Whether a pattern is really interesting or not is a problem
dependent question, and for this reason various similarity measures other than number of co-occurrences have been introduced.
Some of the most common measures are {\em Jaccard}~\cite{journals/tkde/CohenDFGIMUY01},
{\em cosine}, and {\em all\_confidence}~\cite{conf/icdm/LeeKCH03,journals/tkde/Omiecinski03}.
Besides these measures we are also interested
in association rules, which are intimately related to
the {\em overlap coefficient\/} similarity measure.
See~\cite[Chapter 5]{datamining} for background and discussion of similarity measures.

We initiate the study of this problem in the {\em streaming model\/} where 
transactions arrive one by one, and we are allowed limited time per transaction and very small space.
The latter constraint implies we cannot hope to store much
information regarding pairs that are not similar and,
moreover, we cannot store the input.
In particular, classical frequent item set algorithms such as Apriori~\cite{VLDB94*487} and FP-growth~\cite{fp-growth} that work in several passes over the data cannot be used.
The survey of Jiang and Gruenwald~\cite{survey} gives a good overview of the challenges in data stream association mining.

Previous works on transaction data streams
have focused on finding frequent itemsets, and can be
classified in the
following way~\cite{DBLP:conf/vldb/ZhuS02}:
\paragraph*{Landmark model} The frequent itemsets are searched
for in the whole stream, so that itemsets that appeared in the far
past have the same importance as recent ones;
\paragraph*{Damped model} This model is also called
\textit{Time-Fading}. Recent transactions have a higher weight
than the older ones, so nearer itemsets are considered more
interesting than the further;
\paragraph*{Sliding window} Only a part of the stream is
considered at a given time in this model, the one falling
in the sliding window. This implies storing information
concerning the transactions falling within the window,
since whenever a transaction gets out of the window span,
it has to be removed from the counts of the itemsets.

The last two models make the problem of achieving low space usage simpler, since most of the information in the stream has little or no effect on the mining result.
The challenge is instead to handle the real-time requirements of data stream settings.

\medskip
All these approaches look for frequent items and do not
try to compute any similarity, relying on the tacit assumption
that whatever is frequent is automatically interesting.
This assumption is not always true:
\paragraph{Example} Suppose we have item $1$ appearing
in $20\%$ of transactions, item $2$ appearing in $20\%$ of transactions, and the pair $\{ 1, 2 \}$
appears in $10\%$ of transactions. Suppose moreover that
the pair $\{ 3, 4 \}$ appears in only $5\%$ of transactions and
that these transactions are the only ones in which $3$ and
$4$ appear. The set $\{ 1, 2 \}$ has a frequency that is two times
the one of $\{ 3, 4 \}$. But looking at the similarity function
\textit{cosine}, we can easily realize that $s(1,2) =
10 / 20 = 0.5$ while $s(3,4) =  5 / 5 = 1$. If we base the
idea of similarity only on frequencies, we are likely to
miss the pair $\{ 3,4 \}$ which holds a much higher similarity
than the more frequent pair $\{ 1 , 2 \}$.

Notice also that $\{ 3, 4 \}$ holds a higher similarity
for {\em all\/} the measures we are addressing, so the example
shows how frequencies alone do not suffice to infer
similarity properties of pairs.\hfill$\circ$
\medskip

\paragraph{Our contributions}
In this paper we address the problem of finding similar pairs
in a stream of transactions. 
We first show a negative result, which is that a worst-case stream does not allow solutions with non-trivial space usage: 
To approximate even the simplest similarity measure one essentially needs space that would be sufficient to store either the number of occurrences of all pairs or the contents of the stream itself.
Imposing a minimum support $\varphi$ for the items we are interested in alleviates the problem only when $\varphi$ is close to the number of transactions.

\medskip

\begin{theorem}\label{thm:lower}
	Given a constant $k>0$, and integers $m$, $n$, $\varphi$,
	consider inputs of $m$ transactions of total size $mk$ with $n$ distinct items.
	Let $s_{\mathrm{max}}$ denote the highest support among $k$-itemsets where each item has support $\varphi$ or more.
	Any algorithm that makes a single pass over the transactions and estimates $s_{\mathrm{max}}$ within a factor $\alpha < 2$ with error probability $\delta < 1/2$ must use space $\Omega(\min(m,n^k,(m/\varphi)^k))$ bits in expectation on a worst-case input distribution.\hfill $\circ$
\end{theorem}

\medskip

This lower bound extends and strengthens a lower bound for single-item streams presented
in~\cite{journals/tods/CormodeM05}.

Of course, many data streams may not exhibit worst-case behavior. Several papers have considered models of data streams where the items are supposed to be independently chosen from some distribution, or presented in random order~\cite{DBLP:conf/esa/DemaineLM02,journals/isci/YuCLZZ06,conf/stoc/ChakrabartiCM08,Guha:2009:SOO}. 
We present an upper bound that works for a worst-case set of transactions under the condition that it is presented in random order,
which is sufficient to bypass the lower bound. 
Our method is general in the sense that it can evaluate
the similarity of pairs according to several well-established measure
functions.

We note that outside the streaming domain, distributed sorting algorithms, such as the one built into MapReduce, can be used to permute transactions in random order (by using random values as keys).
It seems likely that our approach can also be used in a 1-pass MapReduce implementation.

\medskip
\begin{theorem}\label{thm:upper}
	Let $\delta>0$ be constant, and $s$, $M>1$ be integers.
	We consider a data stream of transactions (subsets of $\{1,\ldots n\}$) of maximum size $M$, where in each prefix the set of transactions appears in random order.
	For all the similarity measures in figure~\ref{fig:measures} there is a streaming algorithm (depending on 
	$s$ and $M$) that maintains a ``$1\pm\delta$ approximation'' of the $s$ most similar high-support pairs in the stream, as follows: 
	Within the $m$ transactions seen so far, let $\Delta$ be the $s$th highest similarity among pairs $\{i,j\}$ where both $i$ and $j$ appear at least $\varphi$ times, where $\varphi$ can be any function of $m$. There exists $L=O(\log(mn))$ such that
	if $\Delta > \frac{L}{\varphi}\max \left\{ \sqrt{\frac{mbM}{s}},M \right\}$, then the pairs maintained all have similarity at least $(1-\delta)\Delta$ with high probability, and all such pairs with similarity $(1+\delta)\Delta$ or more are reported.
	To process a prefix of $mb$ items, the algorithm uses time $O(mb\log(nm))$, with high probability, and space $O(n + s).$\hfill$\circ$
\end{theorem}

\medskip

It is worth noticing that $s$ can be chosen as $O(n)$, which yields a space usage linear in the number of distinct items.
Conversely, choosing $s$ smaller does not improve the space usage, so we may assume $s\geq n$.
In absence of a known bound on the maximum transaction size, one can use $M=n$. 
Then the algorithm guarantees to detect pairs with similarity at least 
$\frac{L}{\varphi}\max \left\{ \sqrt{mb},n \right\}$.
Using $s\geq n$ and ignoring the logarithmic factor $L$ this means that up to input size $mb=n^2$ we can detect similarity $n/\varphi$, and after this point we can detect similarity $\sqrt{mb} / \varphi$.
Assuming that $\varphi$ is chosen as a linear function of $m$ (relative support threshold), we see that the accuracy improves with the length of the stream.

\subsection{Previous work}\label{sec:previous}

Denote by $m$ the number of transactions seen up to
the moment in which we want to report the similar pairs.
Let $n$ indicate the number of distinct items that
can appear in transactions. Without loss of generality we can assume these items
are in the set $\{1,\ldots, n\}$. Parameter $b$ is the average
length of transactions (such that $mb$ is the size of the data set seen so far).

Most of the algorithms we describe actually consider the problem of finding frequent
objects in a stream of items, so they do not focus on itemsets,
like we do. 
But given a stream of transactions we can of course generate the stream of all pairs occurring in these transactions, and feed them to a frequent item algorithm.
(We do not consider here that this might not be possible for large transactions in settings where real-time constraints are important.) 
In the following we let $M_2$ denote the length of the derived stream of pairs.

\subsubsection*{Landmark model}

Many papers have addressed the problem of frequent items
in a stream. Starting from the seminal paper~\cite{journals/jcss/AlonMS99}
streaming algorithms have started to flow in recent years.
Many important contributions to the problem of frequent
items (and indirectly frequent itemsets) have thus been presented.

In several independent papers~\cite{journals/scp/MisraG82,DBLP:conf/esa/DemaineLM02,journals/tods/KarpSP03} 
algorithms have been presented that can find all pairs with support at least $k$ using space $M_2/(k-1)$ and constant time per pair in the stream.
These algorithms may generate false positives, i.e., it is only known that the output will contain the frequent pairs.

Cormode and Muthukrishnan~\cite{journals/tods/CormodeM05} consider the problem of reporting \textit{hot items} in a fully dynamic database scenario.
The space usage is similar to the schemes above, but the error probability can be reduced arbitrarily (at the cost of space).

Also in~\cite{journals/tods/CormodeM05} is a lower bound on the number of bits of memory necessary in order to answer queries that concern reporting the items with
frequencies over a certain threshold. This lower bound is
extended and generalized by our lower bound in theorem~\ref{thm:lower}.

In~\cite{journals/tcs/CharikarCF04} the \textsc{Count Sketch} algorithm
tackles the problem of reporting the $k$ most frequent itemsets. For worst-case distributions their algorithm has similar performance to those mentioned above, but for skewed distributions they are able to detect itemsets with smaller frequencies in the same amount of space.

\subsubsection*{A false negative approach}

Yu et al.~\cite{journals/isci/YuCLZZ06} present algorithms directly addressing
the problem of finding frequent itemsets in a transaction stream. 
The algorithm does not find itemsets that are similar by means of measure
functions other than support. 
Under the assumption that items occur independently (which is arguably quite strong, since we are assuming that there may be dependencies resulting in frequent sets) the authors show upper bounds on space usage similar to those of~\cite{journals/tods/CormodeM05}.
The performance is tested on artificial data sets where the independence assumption holds.
For itemsets of size two (or more) the paper lacks a theoretical analysis of the proposed algorithm, but claims an empirical space usage bounded by $m^3/k^3$.

\subsubsection*{Sampling according to the similarity}

Our algorithms builds on top of an idea presented in~\cite{BiSam,BiSam-jour}.
The sampling technique used in that algorithm is such that pairs
are sampled a number of times that is proportional to their similarity.
(A more technical explanation can be found in section~\ref{sec:pairsampling} where we improve the sampling procedure to make it suitable for a streaming environment.)
The algorithms presented in~\cite{BiSam,BiSam-jour} have near-optimal running time, when no
information on the distribution of similarities are given. As
a matter of fact, the running time is linear in the size of the
input and output (when there are many pairs of roughly the same similarity).
The methods presented are highly general and apply to many
measure functions that are linear in the number
of occurrences of a pair.
However, the method does not directly apply to a streaming setting since it needs two passes over the data.

\section{Lower bound}\label{sec:low}

There are two na\"ive approaches to handling $k$-itemset support counting in a data stream setting: One consists in storing all the transactions seen (possibly trying to compress the representation), and the other one maintains support counts for all $k$-itemsets seen so far.

Theorem~\ref{thm:lower} says that it is not possible to beat the best of these approaches in the worst case (with support threshold $\varphi=1$). The proof is a reduction from communication complexity:

\begin{proof}
	The inputs considered for the lower bound have $m$ transactions of size $k$.
	Let $n'=\min(n,\lfloor mk/(2\varphi)\rfloor)-1$ be the largest possible number of items that can appear $\varphi$ times in $m/2$ transactions, minus~$1$.
	We pick an arbitrary set $F$ of $n'$ items, and will form an input stream that consists of two parts:
	\begin{itemize}
	\item In the first $m/2$ transactions we ensure that each item in $F$ appears $\varphi$ times or more, while no $k$-subset of $F$ appears. This can be done by putting one item not in $F$ in each transaction.
	\item In the last $m/2$ transactions we encode information that will require many bits to store, as detailed below.
	\end{itemize}
	Consider the first $s=\min(m/2,\binom{n'}{k})$ transactions in the second part.
	Since $s\leq \binom{n'}{k}$ we can map the numbers $\{1,\dots,s\}$ to unique $k$-itemsets in $F$.
	In particular, any bit string $x\in\{0,1\}^s$ can be mapped to the unique set of transactions corresponding to the positions of $1$s in $x$.
	In this data set, each $k$-itemset from $F$ appears at most once.

	Suppose we have an algorithm that can determine the support of the most frequent itemset within a factor $\alpha < 2$ with probability $1-\delta$.
	This implies that, on inputs where no itemset appears more than twice, the algorithm can distinguish (with probability $1-\delta$) the cases where the most frequent itemset appears once and twice.
	Given $x\in \{0,1\}^s$ we consider the memory configuration after the algorithm has seen the set of transactions that correspond to $x$.
	This can be seen as a ``message'' that encodes sufficient information on $x$ that allows us to determine if one of the itemsets we have seen appears later in the stream.
	Lower bounds from communication complexity (see~\cite[Example 3.22]{Kushilevitz1997}) tell us that even when we allow error probability $\delta < 1/2$ the amount of communication to determine whether $x,y\in \{0,1\}^s$ have a 1 in the same position (corresponding to the same $k$-itemset appearing twice) is $\Omega(s)$ bits in expectation.
	This means that the memory representation (even if it is compressed) must use $\Omega(s)$ bits.
	Using the estimate $\binom{n'}{k} \geq \min(\binom{n}{k},\binom{mk/(3\varphi)}{k}) = \Omega(\min(n^k,(m/\varphi)^k))$ we get the lower bound stated in the theorem.
\end{proof}

\begin{corollary}
	Any deterministic algorithm that determines the highest support in a transaction data stream must, after having processed transactions of total size $mb$, use space $\Omega(\min(mb,n^k))$ bits on a worst-case input.\hfill$\circ$
\end{corollary}


\section{Our algorithm}\label{sec:algorithm}

We present a new algorithm for extracting similar pairs from a
set of transactions using only one pass over the data. The
algorithm is approximate, so false negatives and false
positives occur.
Most of our discussion will concern space usage, but we are also aiming for very low per-item time complexity of the algorithm.
In particular, we will not allow anything like iterating through all pairs in a transaction.

The measures we will address are reported in Figure~\ref{fig:measures},
and are all symmetric. This means that we are interested only in
looking at pairs $(i,j)$ where $i < j$.
For this reason we will use set notation
for the pairs, so instead of $(i,j)$ we will write $\{i,j\}$.

\paragraph{Parameters of the algorithm}
We recall that $\varphi$ is the item support
threshold, and $M$ is the maximal transaction size.
Increasing $\varphi$ will decrease 
the minimum similarity the algorithm will be able to spot. 
$M$ is a characteristic of the transactions, supplied as a parameter to the algorithm.
In absence of a known bound on $M$, one can set $M=n$.
The parameter $s$ determines the space usage of the algorithm, which is $O(n+s)$ words.

\begin{figure}
\begin{center}
\begin{tabular}{||c|c|c||}
\hline\hline
{\bf Measure} & $s(i,j)$ & $f(|S_i|,|S_j|)$\\
\hline
\hline
{\LARGE\phantom{(}}\textit{Cosine} & $\frac{|S_i\cap S_j|}{\sqrt{|S_i| |S_j|}}$ & $1/\sqrt{|S_i|\cdot |S_j|}$\\
\hline
{\LARGE\phantom{(}}\textit{Dice} & $\frac{|S_i\cap S_j|}{|S_i|+|S_j|}$ & $1/(|S_i|+|S_j|)$\\
\hline
{\LARGE\phantom{(}}\textit{All\_confidence} & $\frac{|S_i\cap S_j|}{\max(|S_i|,|S_j|)}$ & $1 /\max(|S_i|,|S_j|)$\\
\hline
{\LARGE\phantom{(}}\textit{Overlap\_coef} & $\frac{|S_i\cap S_j|}{\min(|S_i|,|S_j|)}$ & $1/\min(|S_i|,|S_j|)$\\
\hline
\hline
\end{tabular}
\caption{\em Measures that we cover with our algorithm and the corresponding functions. The overlap coefficient measure has the property that finding pairs having similarity over a certain threshold implies finding all association rules with confidence over that the same threshold.
As argued in~\cite{BiSam,BiSam-jour}, {\em Jaccard\/} similarity can be handled via dice similarity.
}\label{fig:measures}
\end{center}
\end{figure}

\begin{figure}
\begin{center}
 \begin{tikzpicture}
\node [rectangle,draw] at (0,0) {
\begin{tikzpicture}
\node (x1) [rectangle,draw] at (1,0) {Prefix};
\node [rectangle] at (-.5,.2) {\tiny $T_1, T_2, \ldots, T_m$};
  \draw [-stealth] (-1,0) -- (x1);
\node (x2) [rectangle,draw] at (5,0) {Pair sampling};
  \draw [-stealth] (x1) -- (x2);
\node [rectangle] at (2.7,.2) {\tiny $T_{(m/2)+1}, \ldots, T_m$};
  \draw [->,dashed] (x1) .. controls +(up:1cm) .. (x2);
  \draw [<-,dashed] (x1) .. controls +(down:1cm) .. (x2);
\node [rectangle] at (2.7,-1) {new counts};
\node [rectangle] at (2.7,1.2) {previous counts};
\node (x3) [rectangle,draw] at (5,-3) {\textsc{SampleCount}};
  \draw [-stealth] (x2) -- (x3);
\node [] at (5.5,-.7) {$\{i,j\}$};
\node [] at (5.5,-1.1) {$\{i,p\}$};
\node [] at (5.5,-1.5) {$\{p,q\}$};
\node [] at (5.5,-1.8) {$\vdots$};
\node [] at (5.5,-2.3) {$\{i,j\}$};
\node (res) [circle,thick,draw] at (1,-3) {Similar Pairs};
\draw [-stealth,thick,double] (x3) -- (res);
 \end{tikzpicture}
};
 \end{tikzpicture}
\end{center}
\caption{Overview of the algorithm with all its components.}\label{fig:overview}
\end{figure}
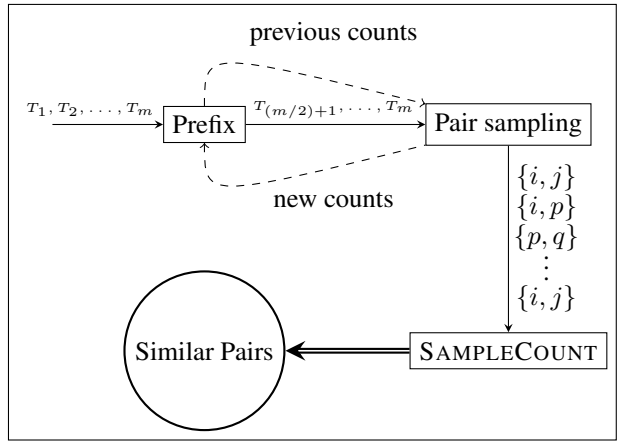

\paragraph{Notation}

In the streaming framework, the total number of transactions is
not known. In order to address this issue, we consider sets of
transactions, \textit{prefixes}, of the stream of increasing size.
Suppose that so far we have seen $m$ transactions $T_1,\ldots,T_m \subseteq \{ 1, \ldots, n \}$.

The \textit{current} prefix has length $2^t$, $t\in {\bf N} \cup 
\{ 0 \}$ when $m$ falls in the
interval $[2^t, 2^{t+1})$.
Our algorithm maintains counts of all items and store copies of the
counts every time the current prefix changes (that is: every time
the number of transactions seen is two times the length of the
current prefix). 
Each time the current prefix changes, we update our estimate of the most similar pairs, and use this estimate until the next change of current prefix.

The algorithm is based on two pipelined stages: a stream of pairs generation
phase and a store and count phase. We will describe the two phases
separately, since the output of the former phase will constitute
the input of the latter. Figure~\ref{fig:overview} gives an overview of the algorithm.

The prefixes of the stream are fed to a {\em pair sampling\/} stage that uses
the stored counts from the previous prefix to compute sampling probabilities. 
Given the current prefix, the
counts relative to that prefix will be used in
order to sample pairs in the stream, until a new set
of counts is stored for the prefix of length $2^{t+1}$
The idea is that, since transactions come in random order, the sampling probabilities should be approximately the same as for the {\sc BiSam} sampling procedure (which bases the sampling probabilities on exact item frequencies).

In section~\ref{sec:analysis} we show how this technique samples,
with high probability, the pairs having a high enough similarity.
In fact, we show that a stronger property holds with high probability: 
Even when we split the stream into $\kappa$ chunks, each with the same number of transactions, we will sample these pairs sufficiently often in each chunk to reliably estimate their similarity.


\subsection{Pair sampling}\label{sec:pairsampling}

We base our technique on the sampling method of the {\sc BiSam} algorithm~\cite{BiSam,BiSam-jour}.
For each transaction the pairs are sampled according to their support, such that the pair $\{i,j\}$ is sampled with probability $\tau f(|S_i|,|S_j|)$, where $f$ is a function that depends on the similarity measure considered, and $\tau$ is a parameter that is used to control the sampling rate.
We fix $\tau=\smash{\frac{4\varphi}{M}}$, where the number of chunks $\kappa$ is given by equation (\ref{eq:kappa}).

\paragraph{BiSam idea}
The idea is that after both $i$ and $j$ have appeared $\varphi$ times, the expected number of times $\{i,j\}$ is sampled is proportional to $s(i,j)$.
Also, the number of samples follows a highly concentrated (binomial) distribution, so the true similarity can be estimated reliably for pairs that are sampled sufficiently often.
For any $f$ that is non-increasing in both parameters, the {\sc BiSam} algorithm performs the sampling in time that is expected linear in the transaction size plus the number of samples.
However, the time to process a transaction may be quadratic with non-negligible probability, which is problematic for application in a streaming context.
We refer to~\cite{BiSam,BiSam-jour} 
for details.

\paragraph{Streaming adaptation}
Two things allow us to arrive at a version suitable for streaming:
\begin{itemize}
\item While {\sc BiSam} produces dependent samples, in the sense that the number of times two different itemsets is sampled is not independent, we show how to make the samples produced independent.
This will ensure that the number of samples from each transaction is highly concentrated around its expectation.
\item The requirement of minimum support $\varphi$ will ensure that processing of a single transaction takes ``linear time with high probability.'' More precisely: Any set of consecutive transactions with a total of $\log m$ items will require linear time with high probability.
\end{itemize}

To achieve independence we will change the sampling probabilities by rounding them down to the nearest negative power of 2.
This means that the expected number of times $\{i,j\}$ is sampled is no longer exactly proportional to $s(i,j)$, but is changed by a factor $\gamma_{i,j} \in [1,2]$.
However, since the sampling probability is known, which means that $\gamma_{i,j}$ will be constant
for any given $\{ i,j \}$, we can still use the sample counts to reliably estimate similarity.

\paragraph{Details}
For a transaction $T_t$ we can visualize the pairs in $T_t\times T_t$ as a 2-dimensional table, with rows and columns sorted by support, where we are interested in the pairs below the diagonal (index $i<j$).
Since $f$ is non-increasing the sampling probabilities are decreasing in each row and column.
This means that for any $k>0$, in time $O(|T_t|)$ we can determine what interval in each row of the table is to be sampled with probability $2^{-k}$.
To produce the part of the sample for one such interval, we describe a method for producing a random sample of $S=\{1,\dots,\phi\}$, for a given integer $\phi$, where each number is sampled with the same probability $p$.
Since $p\phi$ may be much smaller than $\phi$, we want the time to depend on the number of samples, rather than on $\phi$.
This can be achieved using a simple recursive procedure similar to the one used in efficient implementations of reservoir sampling:
With probability $(1-p)^\phi$ we return an empty sample. 
Otherwise, we choose one random element $x$ from~$S$, and recursively take a sample of the set $S\backslash \{x\}$ with sampling probability $p$.
The set $S$ can be maintained in an array, where sampled numbers are marked.
In case more than half of the numbers are marked, we construct a new array containing only unmarked numbers (the amortized cost of this is constant per marking).
To select a random unmarked number we sample until one is found, which takes expected $O(1)$ time because no more than half of the numbers are marked.

In summary, for each sampling probability $2^{-k}$ we can compute the corresponding part of the sample in expected time $O(|T_i|+z_k)$, where $z_k$ is the number of samples.
This is done for $k=1,2,\dots,2\log(nm)$.
Sampling probabilities smaller than $(nm)^{-2}$ are ignored, since the probability that any such pair would be sampled in any transaction is less than $1/m$.
That is, with high probability ignoring such pairs does not influence the sample.
To state our result, let $2^{-{\bf N}}$ denote the set of negative integer powers of 2.
\begin{lemma}
	Let $\tilde{f}: {\bf N} \times {\bf N} \rightarrow 2^{-{\bf N}}$ be non-increasing in both parameters.
	Given a transaction $T_t$ and support counts $|S_i|$ for its items, in expected time $O(|T_t|\log(nm) + z)$ we can produce a random sample of $z$ 2-subsets of $T_t$ such that: 
	\begin{itemize}
	\item $\{i,j\}$ 
	is sampled with probability $\tilde{f}(|S_i|,|S_j|)$ if $\tilde{f}(|S_i|,|S_j|) > (nm)^{-2}$, and otherwise with probability $0$, and
	\item the samples are independent.\hfill$\circ$
	\end{itemize}
\end{lemma}
For all similarity measures in figure~\ref{fig:measures} and any feasible value of $\tau$, the minimum support requirement will ensure that the expected number of samples in a transaction is at most $|T_t|$. 
This means that for each transaction $T_t$, the time spent is $O(|T_t|\log(nm))$ with high probability.

\subsection{SampleCount}\label{sec:samplecount}
This phase sees the stream of pairs generated by the pair sampling,
and has to filter out as many low similarity pairs as possible, while successfully identifying high similarity pairs.
By the properties of pair sampling, this is essentially the task of identifying frequent pairs in the stream of samples.
We aim for space usage that is smaller than that of standard algorithms for frequent item mining in a data stream.
In order to accomplish this we use
a modification of an algorithm by Demaine et al.~\cite{DBLP:conf/esa/DemaineLM02}.
Their algorithm finds frequent items in a randomly permuted stream of items, and so does not directly apply to our setting where only the transactions are assumed to come in random order.
Demaine et al.~are able to sample random elements by simply taking the first elements from the stream.
This would not work in our setting, where all these elements might be pairs coming from the same transaction.

\paragraph{Reservoir sampling}
Instead, we use a
\textit{reservoir sampling} method~\cite{journals/toms/Vitter85}.
We sketch the mechanism here and we refer to the original paper
for a complete description.
Suppose we have a sequence of $d$ items and we want to sample a
random subset of the sequence. We first of all put in the sample
the first $s$ elements that we see. For each subsequent element,
in position $t > s$, we will put it in the sample with probability
$s/t$. When a new element has to be included in the sample,
another one that is already part of the sample has to be evicted.
Each element of the set of samples will be chosen as the victim
with probability $1/s$. This technique ensures we will end up with
a set of samples that is a true random sample of size $s$.

\paragraph{SampleCount}
We consider the stream of pairs divided into $\kappa$ chunks.
The pair sampling generates these chunks such that each chunk corresponds to some set of transactions (i.e., all the pairs sampled from each transaction end up in the same chunk).

We run reservoir sampling on every other chunk to produce a truly random sample of size $s/2$.
We then proceed to count the occurrences of the elements of the sample in the next chunk.
Assume in the following that we number chunks by $[\kappa]$, such that reservoir sampling is done on even-numbered chunks, indexed by $[\kappa_\textrm{even}]$.

When doing the above, whenever we see a pair $\{ i,j \}$ whose count must be updated, we weigh the sample by the factor $\gamma_{i,j}$
that got ``lost'' during the pair sampling phase, so as to consider an expected number of samples
exactly proportional to $s(i,j)$.
At the end of a counting chunk we estimate the similarities of all pairs sampled, and keep the $s/2$ largest similarities seen so far.
At the end of the stream the similarity estimates found are returned to supersede the previous estimates.
Pseudocode for the SampleCount algorithm is shown in
figure~\ref{pseudo:S&C}. 

\begin{algorithm*}
\caption{Pseudocode for the \textsc{SampleCount} phase.}\label{pseudo:S&C}
\begin{algorithmic}[1]
\Procedure{SampleCount}{$P,s,size$}\Comment{$P$ is a stream of pairs, each of which
has associated a similarity value. The \phantom{aaaaaaaaaaaaaaaaaaaaaaaaaaaaaaaaaaaaaaaa}length of $P$ is known.}
  \State $S_\textrm{out} \gets \emptyset$
  \While{There are elements in $P$}\label{while:inf}
  	\State $S' \gets \emptyset$
	\State $S \gets \emptyset$
	\State $t \gets 0$
  	\State $S \gets \mbox{}$ the first $s/2$ elements in $P$
  	\While{$(t < \frac{size}{2} - s/2)$}
		\State $i \gets \mbox{}$ the next element in $P$
		\State Choose uniformly at random a number $r \in [0,1]$
		\If{$r \leq s/(s+2t+2)$}
			\State Choose uniformly at random a victim from $S$ and substitute it with $i$
		\EndIf
		\State $t \gets t + 1$
  	\EndWhile
	\State initialize$(S',S)$\Comment{$S'$ is an associative array indexed on the distinct items 
        present in $S$; initializing it means \phantom{---------------------i--------------}
        putting all its entries to 0}
	\While{$(t < size)$}
		\State $i \gets \mbox{}$ the next element in $P$
		\If{$i \in S$}
			\State $S'(i) \gets S'(i) + \gamma_i$
		\EndIf
		\State $t \gets t+1$
	\EndWhile
	\State Choose the $s$ topmost distinct items between $S'_\textrm{out}$ and $S'$, and assign them to $S'_\textrm{out}$
  \EndWhile
\State Return $S'_\textrm{out}$
\EndProcedure
\medskip
\end{algorithmic}
\end{algorithm*}

\section{Analysis}\label{sec:analysis}
Let $S_i$ denote the set of transactions containing the element
$i$. This means that $S_i \cap S_j$ is the set of transactions
containing the pair $\{ i, j \}$.
Let $S_i^1$ denote the set of transactions containing $i$ in the
current prefix of the stream. Similarly, $S_i^k$ will denote the set
of transactions containing $i$ in $C_k$, the chunk $k$ of the 
suffix of the stream up to the point in which a new current prefix
changes the counts of items occurrences. So $S_i^k = S_i \cap
C_k$

\begin{definition}
Given $x,y\in\mathbf{R}$ we say that $x$ $(\delta,L)$-approximates $y$, written $x\stackrel{_{\delta,L}}{\simeq} y$, if
and only if $x \geq L$ implies $x \in [(1-\delta)y;(1+\delta)y]$.
\hfill $\circ$
\end{definition}
The notation extends in the natural way to approximate inequalities.

In what follows we will use $(\delta,L)$-approximations, where $L=C\log(mn)$ for a suitably large constant $C$ (depending on the accuracy $\delta$ in Theorem~\ref{thm:upper}).
The task is to analyze the accuracy of the new approximation computed when the current prefix changes.
We introduce two random events, \textsc{GoodPermutation} (GP)
and \textsc{GoodBisamSample} (GBS), and bound the probability that they do not happen.

A permutation of the transactions is called \textit{good}
for $\{ i,j \}$, denoted GP$_{i,j}$, if and only if the following conditions hold (for
the current prefix):
\begin{enumerate}
 \item $|S_i^1| \stackrel{_{\delta,L}}{\simeq}|S_i|/2$ and $|S_j^1| \stackrel{_{\delta,L}}{\simeq}|S_j|/2$;\label{eq:eq1}
 \item $\forall k.|S_i^k \cap S_j^k| \stackrel{_{\delta,L}}{\simeq}
 |S_i \cap S_j|/2k$;\label{eq:eq3}
\end{enumerate}
Essentially, goodness means that the frequencies of individual items are close in the first and second half of the current prefix and the frequency of the pair
is evenly spread over the chunks in the second part of the current prefix.

\begin{lemma}\label{lem:GP} Given $\delta \in [0;1] \subseteq
 \mathbf{R}$, we have:
 \begin{displaymath}
  \Pr[\mathrm{GP}_{i,j}] \geq 1 - 6\cdot
e^\frac{-|S_i|\delta^2}{6} 
 \end{displaymath}
\end{lemma}
\begin{proof}
An interesting property of the random variables
$|S_i^1|$ and $|S_i^k \cap S_j^k|$ is that they are negatively dependent~\cite{299634}. 
In a nutshell, the random variables in the vector
$\overrightarrow{\textrm{X}} = (X_1,\ldots, X_n)$ are negatively dependent
if and only if for every two disjoint sets $I,J
\subset \{1,\dots,n\}$ and for every pair of both nonincreasing or both
non decreasing functions $w:\mathbf{R}^{|I|} \mapsto \mathbf{R}$, $g:\mathbf{R}^{|J|} \mapsto \mathbf{R}$,
it holds that $\textrm{E}[f(X_i, i\in I)g(X_j, j\in J)] \leq \textrm{E}[f(X_i, i\in I)]
\textrm{E}[f(X_j, j\in J)]$.
We will use this property later on in the proof.
First of all we bound the probability that $|S_i^1|$ is far from
$|S_i/2|$. Using Chernoff bounds
we can write:
\begin{equation}\label{eq:chB1}
 \Pr[|S_i^1| - |S_i|/2| \leq \delta|S_i| / 2] \leq 2\cdot
  e^{-\frac{|S_i|\delta^2}{6}}
\end{equation}
Looking at $|S_i^k \cap S_j^k|$ we can write:
\begin{equation}\label{eq:chB3}
 \Pr[|S_i^k \cap S_j^k| - |S_i \cap S_j|/2\kappa| \leq \delta
 |S_i \cap S_j| / 2\kappa] \leq 2\cdot e^{-\frac{|S_i \cap S_j|
 \delta^2}{6\kappa}}
\end{equation}
We use the fact that Chernoff bounds also holds for negatively dependent random variables.
Since the last bound is the weakest of the three, the lemma follows.
\end{proof}

We want GP$_{i,j}$ to hold with probability $1-o(1/n^2)$ whenever items $i$ and $j$ both have support $\varphi$.
From Lemma~\ref{lem:GP} we get that this holds if $|S_i\cap S_j| > C\kappa\log n$, for some constant $C$ (depending on $\delta$).
If $s(i,j) > 2\kappa L f(\varphi,\varphi) \geq \kappa L / \varphi$ then 
$|S_i \cap S_j| \geq 2\kappa L$. Hence,
a sufficient condition for the similarity is
\begin{equation}\label{eq:cond1}
 s(i,j) > \kappa L / \varphi \; .
\end{equation}

It remains to understand what is the probability that, given
a good permutation, the pair sampler will take a number
of samples for a given pair in each chunk $k$ that leads to a
$(1\pm\delta)$-approximation of $s(i,j)$.
We denote the latter event by
$\mathrm{GBS}_{i,j,k}$, and want to bound the quantity
$\Pr[\mathrm{GBS}_{i,j,k}|\mathrm{GP}_{i,j}]$.


For this purpose consider the random variable
$X_{i,j,k}$ defined as the number of times we sample the pair $\{ i, j \}$ in
chunk~$k$.
Assuming $\mathrm{GP}_{i,j}$
we have that (over the randomness in the pair sampling algorithm) $\E[X_{i,j,k}] \stackrel{_{\delta,L}}{\simeq} \tilde{f}(|S_i^1|,|S_j^1|)
\tau |S_i \cap S_j| / 2\kappa$.
Since the occurrences of $\{ i,j \}$ are independently sampled, we can apply a
Chernoff bound to conclude $X_{i,j,k} \stackrel{_{\delta,L}}{\simeq}\E[X_{i,j,k}]$.
%
This leads to the conclusion:
\begin{lemma}\label{lem:GBS}
 $X_{i,j,k}\stackrel{_{\delta,L}}{\simeq}\tilde{f}(|S_i^1|,|S_j^1|) \tau
 |S_i \cap S_j| / 2\kappa$\hfill $\circ$
\end{lemma}

Suppose that $X_{i,j,k}$ is close to its expectation. Then we can use it, with $(1\pm\delta)$-approximations of $|S_i|$ and $|S_j|$, to compute a $(1\pm O(\delta))$-approximation of $s(i,j)$. This follows by analysis of the concrete functions $f$ of the measures in Figure~\ref{fig:measures}.


A sufficient condition on the similarity needed for a $(1\pm\delta)$-approximation of $X_{i,j,k}$ can be inferred from lemma~\ref{lem:GBS}.
If $s(i,j) \geq 4\kappa L / \tau$ then $\E[X_{i,j,k}]
\geq s(i,j) \tau / 4 \kappa \geq L$. 
So it suffices to enforce:
\begin{equation}\label{eq:cond2}
  s(i,j) \geq 4 \kappa L / \tau \;.
\end{equation}

In order to have $O(mb)$ pairs produced
by the pair sampling phase, we will choose $\tau = 4 \varphi / M$.
The expected number of pair samples from $T_t$ is less than $|T_t|^2 \tau f(\varphi,\varphi)$, using that $f$ is decreasing.
For all measures we consider, $f(\varphi,\varphi)\leq 1/\varphi$, so
$|T_t|^2 \tau f(\varphi,\varphi) \leq |T_t|^2 / M \leq |T_t|$.

It remains to understand which is the probability that a pair of items,
each with support at least $\varphi$, is not sampled by SampleCount.
Let the random variable $X_{.,.,k}$ represent the total number of
samples taken in chunk $k$. The probability that a $\{i,j\}$ is sampled in chunk $k$ is $X_{i,j,k} / X_{.,.,k}$, so the
probability that it does not get sampled in any (even-numbered) chunk is
$\prod_{k\in [\kappa_\textrm{even}]} (1 - X_{i,j,k} / X_{.,.,k})^s$.
We have seen before that $X_{i,j,k} \stackrel{_{\delta,L}}{\geq} s(i,j)\tau /4 \kappa$.
For what concerns $X_{.,.,k}$ using a Chernoff bound we can get:
$X_{.,.,k} \stackrel{_{\delta,L}}{\simeq} \E[X_{.,.,k}] \leq mb / \kappa$,
using the linear upper bound on the number of samples.
So we can compute:
\begin{alignat*}{1}
 & \prod_{k\in [\kappa_\textrm{even}]} (1 - X_{i,j,k} / X_{.,.,k})^s \leq
 \left(1 - \frac{s(i,j) \tau \kappa}{2 \kappa \gamma_{i,j} mb}\right)^{s\kappa/2} \\
 & \leq \left(1 - \frac{s(i,j) \tau}{4mb}\right)^{s\kappa/2} \leq C \exp\left[-\frac{s(i,j) \tau s \kappa}{8mb}\right]
\end{alignat*}
In order for this probability to be small enough ($O(1/m^2)$),
we need to bound the similarity to
\begin{equation}\label{eq:cond3}
 s(i,j) \geq \frac{8mbL}{s\kappa\tau}
\end{equation}

To choose the best value of $\kappa$ we balance constraints~(\ref{eq:cond1})
and~(\ref{eq:cond3}), getting:
\begin{equation}\label{eq:kappa}
 \frac{\kappa L}{\varphi} = \frac{mbL}{s\kappa\tau} \Rightarrow \kappa  = \sqrt{\frac{mbM}{s}}
\end{equation}
From which we can deduce:
\begin{equation}
 s(i,j) = \frac{L}{\varphi}\max \left\{ \sqrt{\frac{mbM}{s}},M \right\} \; .
\end{equation}

\section{Dataset characteristics}\label{sec:experiments}

We have computed, for a selection of the datasets hosted on the
FIMI web page\footnote{\texttt{http://fimi.cs.helsinki.fi/}},
the ratios between the number of occurrences
of single items and pairs in the first half of the transactions and
the total number of occurrences of the same items or pairs.
The values of some of this ratios, the most representative,
are plotted figure~\ref{fig:plot};
on the $x$-axis items or pairs are spread evenly, after they
have been sorted according to their associated ratio. The
$y$-axis represents the value of the ratios.
We have taken into account only items and pairs whose support
is over $20$ occurrences in the whole dataset, in order to
avoid the noise that could be generated by very rare elements.
\begin{figure*}
\centering
\subfigure
{\includegraphics[width=.4\textwidth]{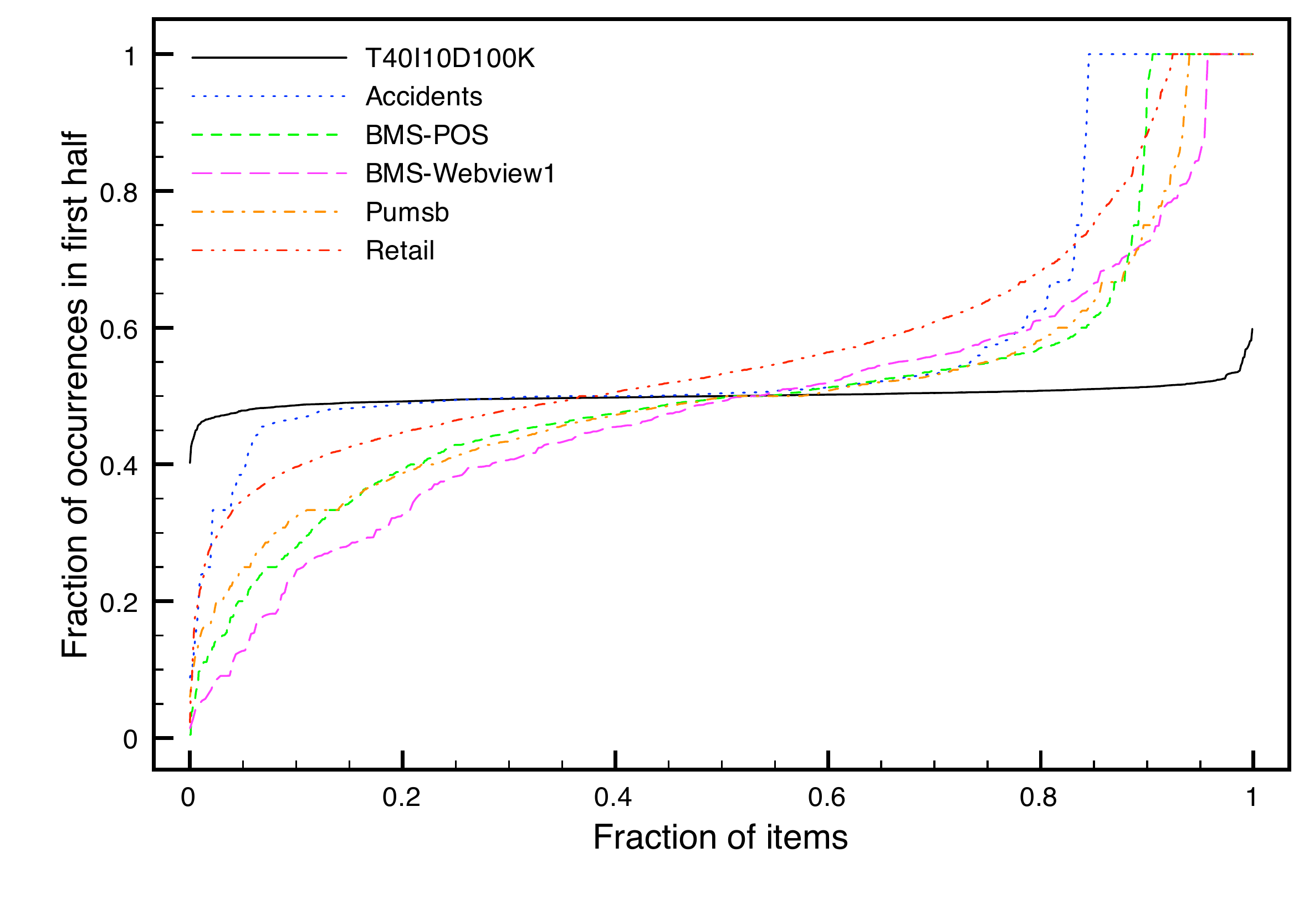}}
\hspace{5mm}
\subfigure
{\includegraphics[width=.4\textwidth]{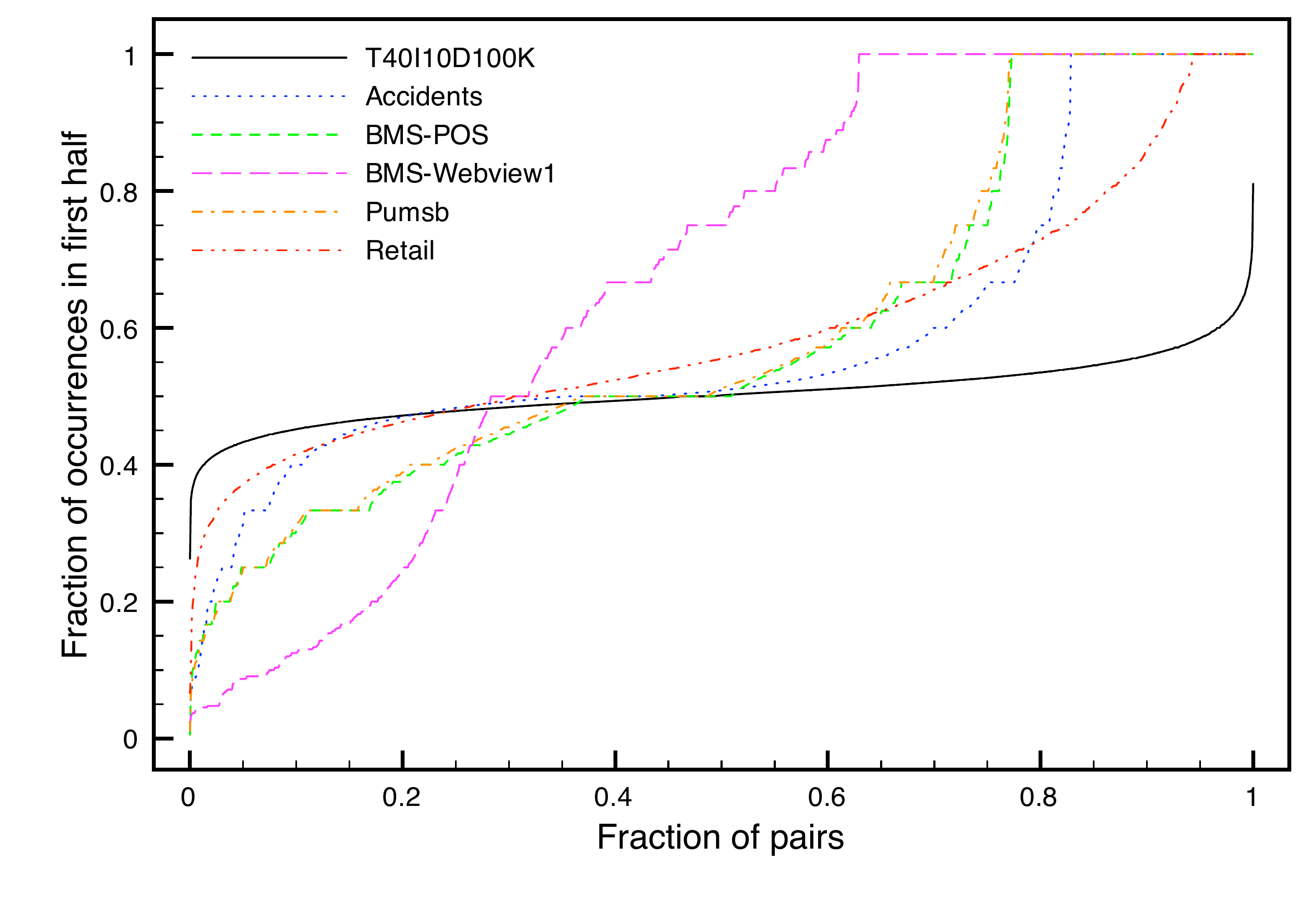}}
\caption{Plots of the ratios $|S_i^1|/|S_i|$ and $|S_i^1 \cap S_j^1| /|S_i \cap S_j|$.}\label{fig:plot}
\end{figure*}
As we can see, the number of occurrences and co-occurrences
are not so far from what would be expected under a random permutation of the transactions.
The synthetic data set behaves exactly like we would expect under a random permutation, with the ratio being very close to $1/2$ for almost all items/pairs.

This means that even for real data sets, where the order of transactions is not random,
the sampling probabilities used in the pair sampling are reasonably close to the ones that would be obtained under the random permutation assumption.

\section{Conclusions}
We presented the first study concerning the problem of mining similar
pairs from a stream of transactions that does rely on the similarity
of items and not only on the frequency of pairs. 
The structure of the
problem is studied and exploited in order to highlight a result of non
possibility and show a suitable algorithm that is fast and space-efficient.
A thorough experimental study of (carefully engineered versions of) the presented algorithm remains to
be carried out.

An interesting open question is to extend the lower bound presented in section~\ref{sec:low}
to our case of study, in which the transactions are given in random order.
\bibliographystyle{plain}
\bibliography{mining}

\end{document}